\documentclass[runningheads]{llncs}
\usepackage{graphicx}
\usepackage{amssymb,amsfonts,amsmath}
\usepackage{booktabs}
\usepackage{cancel}
\usepackage{color}
\usepackage{hyperref}
\newcommand{\state}[1]{\mathbf{s}_{#1}}
\newcommand{\KL}[2]{D_{\text{KL}}\left( #1 \| #2 \right)}
\newcommand{\expect}[2]{\mathbb{E}_{#1} \left[ #2 \right]}

\newcommand{\expon}[1]{\exp{\left( #1 \right)}}

\begin{document}
\title{Deriving time-averaged active inference from control principles}
%
%
\author{Eli Sennesh\inst{1}\orcidID{0000-0001-7014-8471} \and
Jordan Theriault\inst{1}\orcidID{0000-0002-4680-0172} \and
Jan-Willem van de Meent\inst{2}\orcidID{0000-0001-9465-5398} \and
Lisa Feldman Barrett\inst{1}\orcidID{0000-0003-4478-2051} \and
Karen Quigley\inst{1}\orcidID{0000-0001-8844-990X}
}

\authorrunning{E. Sennesh et al.}
%
\institute{Northeastern University, Boston MA 02115, USA \\ 
\email{\{sennesh.e,jordan\_theriault,l.barrett,k.quigley\}@northeastern.edu}
\and
University of Amsterdam, 1090 GH Amsterdam, the Netherlands \\
\email{j.w.vandemeent@uva.nl}
}

\maketitle              
\begin{abstract}
Active inference offers a principled account of behavior as minimizing average sensory surprise over time. Applications of active inference to control problems have heretofore tended to focus on finite-horizon or discounted-surprise problems, despite deriving from the infinite-horizon, average-surprise imperative of the free-energy principle. Here we derive an infinite-horizon, average-surprise formulation of active inference from optimal control principles. Our formulation returns to the roots of active inference in neuroanatomy and neurophysiology, formally reconnecting active inference to optimal feedback control. Our formulation provides a unified objective functional for sensorimotor control and allows for reference states to vary over time.
\end{abstract}

\vspace{-1em}
\section{Introduction}
\vspace{-0.5em}
\label{sec:introduction}
Adaptive action requires the integration and close coordination of sensory with motor signals in the nervous system. Active inference~\cite{friston2017active} provides one of the few available unifying theories of sensorimotor control; it says that the nervous system encodes both sensory and motor signals as afferent predictions and reafferent prediction errors. Sensory predictions induce errors that can only be quashed by updating the predictions, while motoric predictions induce errors that can be quashed by simply moving the body to conform to the predicted trajectory \cite{adams2013predictions}. The free energy principle, following the logic of active inference, says that organisms maintain their self-organization as a whole over time by avoiding surprising interactions between their internal and external environments \cite{friston2010free}. This entails maintaining bodily states within homeostatic ranges~\cite{pezzulo2015active} by issuing sensory, proprioceptive, and interoceptive predictions that minimize errors under a ``prior preference''~\cite{da2020active} or ``non-equilibrium steady-state''~\cite{friston2010generalised} density. Such a density must be stationary throughout time.

Early ``non-equilibrium steady-state'' formulations of active inference provided probability densities over full trajectories of movement and interaction \cite{friston2010generalised,friston2009reinforcement}. In regulatory terms, this corresponds to covariation of bodily states under a ``just enough, just in time''~\cite{sterling2012allostasis} mode of regulation that physiologists have labelled homeostasis~\cite{carpenter2004homeostasis} with time-varying set points, rheostasis~\cite{mrosovsky1990rheostasis}, and recently allostasis~\cite{sterling2012allostasis,corcoran2019allostasis,schulkin2019allostasis,Tschantz2022}. A control theorist would call these trajectories or set-points a \emph{reference trajectory} or ``reference signal'' that a controller tries to track. However, many more recent formulations of active inference use state-space models with fixed ``prior preferences'' that correspond to homeostatic set-points or ranges \cite{da2020active}. They also typically employ either finite time horizons or exponential discounting of expected free energy, unlike the original formulation of active inference in terms of average surprise over time. A control theorist would refer to these as reference states rather than reference trajectories.

This paper will rederive active inference as minimization of path-entropy over an infinite time horizon. The paper's formulation will derive from the first principles of infinite-horizon, average-cost optimal control; will allow preferences to vary according to their own generative model, and will unify motor active inference~\cite{adams2013predictions} (mAI) with decision active inference~\cite{Smith2022} (dAI). This will also unify the computational principles behind motor active inference - the ``equilibrium point''~\cite{Feldman1986,Latash2010} or ``reference configuration''~\cite{Feldman2015} hypotheses - with the higher-level study of sensorimotor behavior as optimal feedback control. Finally, the paper's formalism will provide a unified free energy functional for perception, motor action, and decision making over time.

Section~\ref{sec:notations} will explain this paper's notation and lay out an example generative model supporting the necessary features for the intended formulation of active inference. Section~\ref{sec:fep} will summarize belief updating in generative models, give a recognition model to match the generative model, describe the free energy principle for perceptual inference, and finish by describing active inference. Section~\ref{sec:feedback_pic} will then extend active inference to the setting of an explicit reference model prescribing behavior and give the control criterion corresponding to active inference under the free energy principle. Section~\ref{sec:time_averaged_pic} will derive the resulting free energy bounds whose optimization will yield a Bellman-optimal feedback controller based on the generative and recognition models. Section~\ref{sec:apic_summary} will discuss related work; consider implementation issues for infinite-horizon, average-cost active inference; and conclude. Appendix~\ref{app:derivations} will provide derivations for equations that would otherwise have broken the flow of the paper.

\vspace{-1em}
\section{Preliminaries and notation}
\vspace{-0.5em}
\label{sec:notations}

\begin{figure}[t]
    \centering
    \includegraphics[width=0.6\columnwidth]{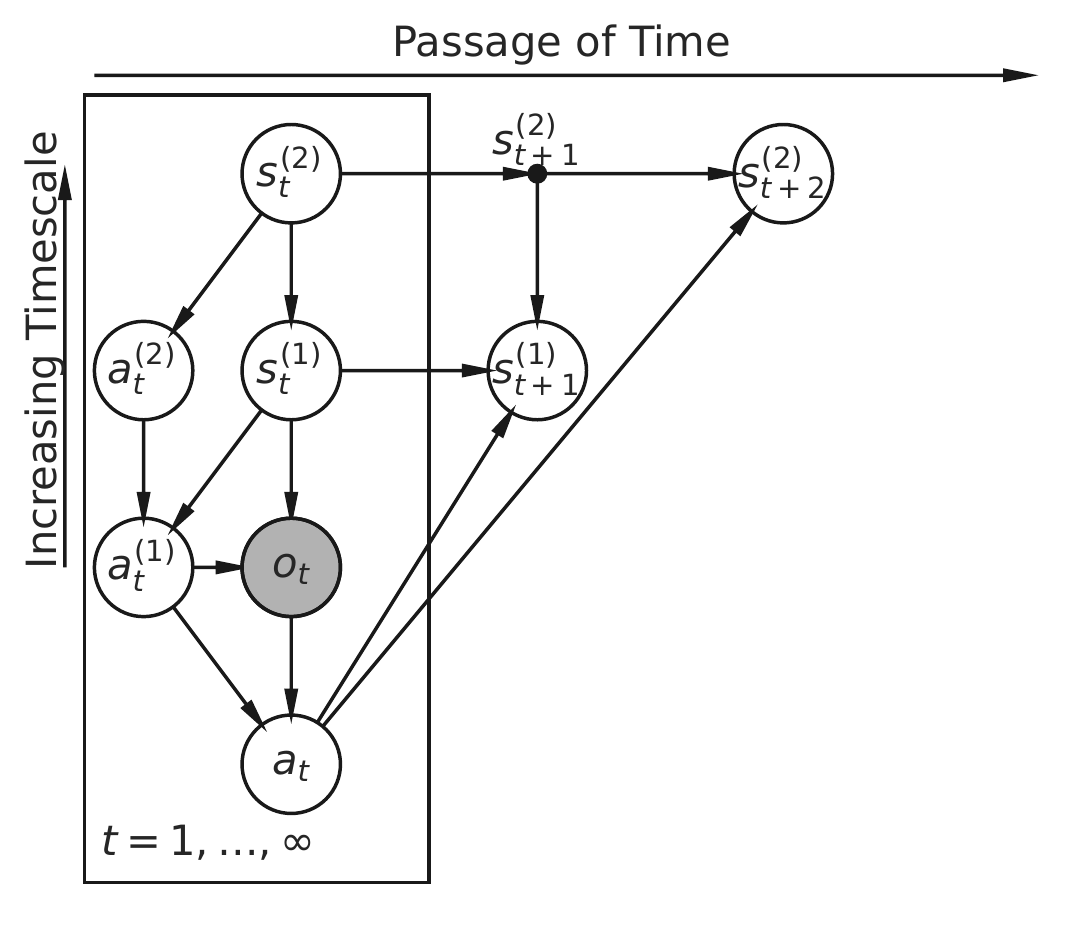}
    \caption{A hierarchical generative model we use as an example in this paper.  Two variables \((s^{(1)}_{t}, s^{(2)}_{t})\) denote unobserved latent states, and each $a^{(k+1)}_t$ parameterizes a reference model for $s^{(k)}_t$.  $o_t$ represents observed sensory outcomes, and $a_t$ represents the feedback control actions generated by motor reflex-arcs.}
    \label{fig:allostasis_pgm}
    \vspace{-1em}
\end{figure}

This paper will explain its formulation of active inference in terms of the discrete-time graphical model in Figure~\ref{fig:allostasis_pgm}. Like many generative models used to lay out active inference~\cite{Kiebel2008,Pezzulo2018}, this model employs a hierarchy of temporal scales. We number these timescales from the shortest to the longest, while numbering random variables with discrete timesteps $t \in 1\ldots T$ from left to right. For simplicity, we also restrict our graphical model to only three levels of hierarchy: observable variables, fast latent variables, and slow random variables. Following those rules, observations $o_t$ and feedback motor actions $a_t$ are 1-Markov; they ``tick'' at every time-step. The fast latent variables $s^{(1)}_t$ and $a^{(1)}_t$ also change at every time-step.  At the next level up, slow latent variables $s^{(2)}_t$ and $a^{(2)}_t$ are 2-Markov; they ``tick'' every second time-step $t+2$. We assume arbitrary state spaces for all random variables, latent and observed, without any discrete or linear-Gaussian assumptions about their conditional densities. Some evidence suggests~\cite{howard2022formal} that the brain may in fact represent time by learning a combination of frequencies in the Laplace domain~\cite{shankar2012scale}, and so the use of only three levels in the model should not be taken to describe anything biological.

We write the combined latent states
\begin{align*}
    s^{(1:2)}_t &= (s^{(1)}_t, s^{(2)}_t)
\end{align*}
and the ``actions'' or reference states
\begin{align*}
    a^{(0:2)}_t &= (a_t, a^{(1)}_t, a^{(2)}_t).
\end{align*}
We can therefore write the complete state at a time-step $t$ as
\begin{align*}
    \state{t} &= (o_t, s^{(1:2)}_t, a^{(0:2)}_{t}).
\end{align*}
We will denote probability densities over actions as policies $\pi$ and probability densities in the generative model as $p_\theta$ (with arbitrary parameters $\theta$). The lowest level of conditional probability densities then consists of
\begin{align*}
    p_{\theta}(a_t, o_t \mid a^{(1)}_t, s^{(1)}_t) &= \pi(a_t \mid o_t, a^{(1)}_t) p_{\theta}(o_t \mid a^{(1)}_{t}, s^{(1)}_{t}),
\end{align*}
the fast latent state level consists of
\begin{align*}
    p_{\theta}(a^{(1)}_t, s^{(1)}_t \mid s^{(1)}_{t-1}, a^{(2)}_{t}, s^{(2)}_{t}, a_{t-1}) &= \pi(a^{(1)}_t \mid s^{(1)}_{t}, a^{(2)}_t) p_{\theta}(s^{(1)}_t \mid s^{(1)}_{t-1}, s^{(2)}_{t}, a_{t-1}),
\end{align*}
and the slow latent state level consists of
\begin{align*}
    p_{\theta}(a^{(2)}_t, s^{(2)}_t \mid s^{(2)}_{t-1}, a_{t-1}) &= \pi(a^{(2)}_t \mid s^{(2)}_t) p_{\theta}(s^{(2)}_t \mid s^{(2)}_{t-1}, a_{t-1}).
\end{align*}
We write the complete state of the \emph{generative model} $p_\theta$ at a time-step \(t\) with its associated conditional densities as
\begin{multline}
    \label{eq:state_conditional}
    p_{\theta}(\state{t} \mid \state{t-1}) = p_{\theta}(a_t, o_t \mid a^{(1)}_t, s^{(1)}_t) p_{\theta}(a^{(1)}_t, s^{(1)}_t \mid s^{(1)}_{t-1}, a^{(2)}_{t}, s^{(2)}_{t}, a_{t-1})
    \\
    p_{\theta}(a^{(2)}_t, s^{(2)}_t \mid s^{(2)}_{t-1}, a_{t-1}),
\end{multline}
and the joint density over time (conditioned on a fixed initial state $\state{0}$) as
\begin{align}
    \label{eq:state_joint}
    p_\theta(\state{1:T} \mid \state{0}) &= \prod_{t=1}^{T} p_{\theta}(\state{t} \mid \state{t-1}).
\end{align}
The model treats outcomes $o_t$ as observed, $a_t$ as a feedback-driven motor action, and other variables as latent.
Inspired by the referent configuration account of motor control~\cite{Feldman2015,Latash2019}, the model treats $a^{(1:2)}_t$ as parameterizing ``prior preferences'' or referent configurations
\begin{align}
    \label{eq:reference_distribution}
    R(\state{t}) &= R(o_t \mid a^{(1)}_t) R(s^{(1)}_t \mid a^{(2)}_{t}).
\end{align}
$a_t$ models the feedback control action of motor reflexes. $a^{(1)}_t$ parameterizes a reference state for $o_t$. $a^{(2)}_t$ parameterizes a reference model for $s^{(1)}_t$. Since reference trajectories direct action, we consider their distributions to be policies
\begin{align}
    \label{eq:policy}
    \pi(a_t, a^{(1:2)}_t \mid o_t, s^{(1:2)}_t) &= \pi(a_t \mid a^{(1)}_t, o_t) \pi(a^{(1)}_t \mid s^{(1)}_{t}, a^{(2)}_t) \pi(a^{(2)}_t \mid s^{(2)}_t).
\end{align}
$s^{(2)}_t$, as the highest level latent state, has no reference model. In neuroscience, it might correspond
to predictive modeling at the highest level of the neuraxis or cortical hierarchy \cite{Barrett2015a,Livneh2020,Quigley2021}. In an engineering setting, it might contain both environment and task state~\cite{Nasiriany2019,ringstrom2020jump,Tang2021} or reward machine~\cite{Icarte2018,Camacho2019} state.

The likelihood \(p_{\theta}(o_t \mid a^{(1)}_t, s^{(1)}_t)\) does not specify the reference model; it instead provides the statistical grounding for both the latent states and the reference model parameters. The model here does not assume that reference densities at all levels are prespecified or learned, but instead leaves that issue open.

\begin{table}[t]
    \centering
    \begin{tabular}{||c|l||}
        \hline
        $p_\theta$ & Probability density for the generative model \\
        $q_\phi$ & Probability density for the recognition model \\
        $R$ & Probability density for the reference model \\
        $\pi$ & Policy density over actions and references \\
        $t$ & Discrete time-step index \\
        $o_t$ & Observations \\
        $s^{(1:2)}_t$ & Unobserved model states \\
        $a^{(1:2)}_t$ & Parameters to a reference model $R$ \\
        $a_t$ & Control actions \\
        $\state{t}$ & A complete model state for time $t$ \\
        \hline
    \end{tabular}
    \vspace{0.5em}
    \caption{Random variable names used in this paper}
    \label{tab:apic_random_variables}
    \vspace{-1em}
\end{table}

We then designate as cost functions the surprisals over complete states (under the reference model) and over observations (under the generative model)
\begin{align}
    \label{eq:allostatic_objective}
    J(\state{t}) &= -\log R(\state{t}), \\
    \label{eq:likelihood_objective}
    L(\state{t}) &= -\log p_{\theta}(o_t \mid a^{(1)}_t, s^{(1)}_t).
\end{align}
Equation~\ref{eq:allostatic_objective} equals the negative of the reward function used in distribution-conditioned reinforcement learning \cite{Nasiriany2021} and can represent any control objective.

This paper will condition behavioral trajectories upon an initial state \(\state{0}\) as context. This initial state corresponds to the beginning of a behavioral episode. The following states from time $1$ until time $T$, sampled from a generative model with parameters $\theta$, are then written as sampled from the joint density
\begin{align*}
    \state{1:T} &\sim p_\theta(\state{1:T} \mid \state{0}).
\end{align*}
This section has described a generative model and a decision objective under which to formulate active inference. Table~\ref{tab:apic_random_variables} summarizes the notation the rest of the paper will use. The next section will lay out belief updating for the generative model, a recognition model to represent updated beliefs, and the free energy principle for perceptual inference. Later sections will show how to extend free-energy minimization to approximate a feedforward planner (in the generative model) and feedback controller (in the recognition model) that minimize surprisal under the reference model.

\vspace{-1em}
\section{Surprise minimization and the free energy principle}
\vspace{-0.5em}
\label{sec:fep}

Section~\ref{sec:notations} gave a generative model and a way of writing arbitrary preferences as probability densities. However, the formalism constructed so far would induce a merely feedforward model-based planner, one which could not correct upcoming movements in light of observations. Bayes' rule specifies how to update probabilistic beliefs about unobserved variables in light of observations:
\begin{align}
    \label{eq:bayes}
    p_\theta(s^{(1:2)}_{1:t}, a^{(1:2)}_{1:t} \mid o_{1:t}, \state{0}) &= \frac{p_\theta(o_{1:t}, s^{(1:2)}_{1:t}, a^{(1:2)}_{1:t} \mid \state{0})}{
        p_\theta(o_{1:t} \mid \state{0})
    }.
\end{align}
The denominator of Equation~\ref{eq:bayes} is called the marginal likelihood, and its negative logarithm is the surprise under the generative model
\begin{align*}
    h(o_{1:t}) &= -\log p_\theta(o_{1:t} \mid \state{0}).
\end{align*}
Friston's free energy principle~\cite{Friston2010} posits that a system, organism, or agent in a changing environment preserves its structure against the randomness of its environment by embodying a generative model of its environment and minimizing that model's long-term average surprise
\begin{align}
    \label{eq:surprise}
    H(o_t) &= \lim_{T\rightarrow\infty} \frac{1}{T} -\log p_\theta(o_{1:T} \mid \state{0}).
\end{align}
In most generative models, neither the denominator of Equation~\ref{eq:bayes} nor the surprise of Equation~\ref{eq:surprise} are analytically tractable, and Bayesian inference requires approximation. Active inference in particular approximates optimal belief updating by substituting a tractable \emph{recognition model} $q_\phi$ (with parameters $\phi$) for the posterior distribution
\begin{align*}
    s^{(1:2)}_{1:T}, a^{(1:2)}_{1:T} &\sim q_\phi(s^{(1:2)}_{1:T}, a^{(1:2)}_{1:T} \mid o_{1:T}, a_{1:T}, \state{0}), \\
    q_\phi(s^{(1:2)}_{1:T}, a^{(1:2)}_{1:T} \mid o_{1:T}, a_{1:T}, \state{0}) &= \prod_{t=1}^{T} q_\phi(s^{(1:2)}_{t}, a^{(1:2)}_{t} \mid o_{t}, a_{t}, \state{t+1}, \state{t-1}).
\end{align*}
This recognition model is conditioned on both the previous time-step $t-1$ and the next time-step $t+1$, and can therefore perform retroactive belief updates.

To improve the recognition model's approximation to the posterior distribution, active inference entails evaluating and minimizing the \emph{variational free energy} (Equation~\ref{eq:vfe}, derivation in Proposition~\ref{prop:vfe} in Appendix~\ref{app:derivations})
\begin{multline}
    \label{eq:vfe}
    \mathcal{F}_{\theta, \phi}(t) = \expect{q_\phi}{- \log p_{\theta}(o_t \mid a^{(1)}_t, s^{(1)}_t)} + \\ \KL{q_\phi(s^{(1:2)}_{t}, a^{(1:2)}_{t} \mid o_{t}, \state{t+1}, \state{t-1})}{p_{\theta}(s^{(1:2)}_t, a^{(1:2)}_t \mid \state{t-1})}.
\end{multline}
The free energy serves as a tractable upper bound to the surprise
\begin{align*}
    H(o_t) &\leq \mathcal{F}_{\theta,\phi}(t).
\end{align*}
Intuitively, given an observation at each time-step $t$, minimizing the free energy amounts to updating the beliefs of the recognition model $q_\phi$ to approximate the posterior distribution of the generative model $p_\theta$. A model-based controller can then use those updated beliefs to revise or plan actions into the future. Active inference has therefore often been formulated as using action to minimize this free energy bound. Such a move then prompts the question of how to encode a desirable reference trajectory into the generative model or another term of the free energy bound \cite{Friston2012}. The next section will define notions of surprise and free energy that encode fit to an explicitly specified reference trajectory.

\section{Active inference with an explicit reference}
\label{sec:feedback_pic}
Minimizing free energy fits a model-based controller's generative and recognition models to ongoing trajectories of observations. However, for the updated beliefs to determine action, the controller must use them to evaluate the fit to the reference trajectory (Equation~\ref{eq:allostatic_objective}) and emit motor actions. Fortunately, Thijssen~\cite{Thijssen2015} gave  an interpretation of probabilistic updating in terms of control: the recognition model $q_\phi$ acts as a \emph{state-feedback controller}, for which the variational free energy becomes a running control cost. This section will show how to evaluate fit to the reference trajectory under the recognition model, and specify the functional it must optimize to serve as a feedback controller.

The generative model in Section~\ref{sec:notations} and recognition model in Section~\ref{sec:fep} use discrete time-steps and explicitly specify the ``pathwise'' reference model separately from the generative and recognition models. The surprise to minimize is therefore the long run average of the cross-entropy
\begin{align}
    \label{eq:ref_cross_entropy}
    H(q_\phi, R) &= \lim_{T\rightarrow\infty} \frac{1}{T} \sum_{t=1}^{T} \expect{\state{t} \sim q_\phi}{-\log R(\state{t})}.
\end{align}
Equation~\ref{eq:ref_cross_entropy} gives the long-term average surprise of using the reference model to approximate the posterior beliefs of the recognition model. Replacing the reference model with the forward generative model would then amount to minimizing the long-term average surprise (entropy); this generalization treats the reference model as specifying a trajectory for the feedback controller to track.

Standard properties of free energy functionals imply that a desirable objective functional would upper bound the sum of reference surprise and sensory surprise
\begin{align}
    \label{eq:surprisal_bound}
    H(R(\state{t})) + H(o_t) &\leq \mathcal{J}(t).
\end{align}
Such a free energy functional would balance the reference model's surprise (the first term) with the generative model's surprise (the second term). In fact it can be formed simply by adding Equation~\ref{eq:ref_cross_entropy} to Equation~\ref{eq:vfe}
\begin{align}
    \label{eq:pic_objective}
    \mathcal{J}_{\theta,\phi}(t) &= H(q_\phi, R) + \mathcal{F}_{\theta,\phi}(t) \\
    &= \expect{\state{t} \sim q_\phi}{J(\state{t})} + \mathcal{F}_{\theta,\phi}(t),
\end{align}
and expanding the term for Equation~\ref{eq:vfe} will yield a long-form expression
\begin{multline}
    \label{eq:pic_objective_long}
    \mathcal{J}_{\theta,\phi}(t) = \expect{q_\phi}{J(\state{t})} + \expect{q_\phi}{L(\state{t})} + \\
    \KL{q_\phi(s^{(1:2)}_{t}, a^{(1:2)}_{t} \mid o_{t}, \state{t+1}, \state{t-1})}{p_{\theta}(s^{(1:2)}_t, a^{(1:2)}_t \mid \state{t-1})}.
\end{multline}
Equation~\ref{eq:pic_objective_long} gives an objective functional in terms of
\begin{itemize}
    \item The reference surprisal under the recognition model,
    \item The observation surprisal under the recognition model, and
    \item The deviation of the recognition model from the generative model.
\end{itemize}

Neuroscientists~\cite{daw2000behavioral,shadmehr2020vigor} and ecologists~\cite{stephens2019foraging} have found evidence that animals optimize a \emph{global capture rate} $\Bar{\mathcal{J}}$ in many decisions: rewards minus costs, divided by time. Active inference modelers typically ground the construct of ``reward'' in reduction of surprise \cite{morville2018homeostatic}, and so a broad field of evidence comes together to support the time-averaging functional form implied by Bayesian mechanics in both their ``steady-state density'' and ``pathwise'' formulations \cite{ramstead2022bayesian}. The next section will therefore apply the principles of stochastic optimal feedback control for the partially observed setting and \emph{average-cost criterion}, and solve the resulting control problem to formulate active inference.

\section{Deriving time-averaged active inference from optimal control}
\label{sec:time_averaged_pic}

The average-cost criterion for optimality entails minimizing the \emph{indefinite} surprise rate with respect to the generative model $p_{\theta}(\state{1:T} \mid \state{0})$
\begin{align}
    \label{eq:indefinite_capture_rate}
    \Tilde{\mathcal{J}}(\state{0}) &= \lim_{T\rightarrow\infty} \expect{p_\theta(\state{1:T} \mid \state{0})}{
        \Bar{\mathcal{J}}_{\theta,\phi}(\state{1:T})
    }.
\end{align}
This minimization requires estimating Equation~\ref{eq:indefinite_capture_rate} for each behavioral episode in context, a ``global surprise rate'' in terms of $\mathcal{J}_{\theta,\phi}(t)$
\begin{align}
    \label{eq:global_capture_rate}
    \Bar{\mathcal{J}}_{\theta,\phi}(\state{1:T}) &= \frac{1}{T} \sum_{t=1}^{T} \mathcal{J}_{\theta,\phi}(t).
\end{align}
Plugging Equation~\ref{eq:pic_objective_long} into Inequality~\ref{eq:surprisal_bound} shows that minimizing Equation~\ref{eq:global_capture_rate} will, by proxy, minimize the reference and sensory surprise in the context of a sampled state trajectory $\state{0:T}$. This estimation does not require a prespecified episode length $T$, and can be performed under the generative model
\begin{align}
    \label{eq:average_capture_rate}
    \Bar{\mathcal{J}}_{\theta,\phi}(\state{0}) &= \expect{\state{1:T}\sim p_{\theta}(\state{1:T} \mid \state{0})}{
        \Bar{\mathcal{J}}_{\theta,\phi}(\state{1:T})
    }.
\end{align}
Having estimates of Equation~\ref{eq:average_capture_rate} will enable minimizing the mean-centered surprise at each time-step
\begin{align}
    \label{eq:advantage_function}
    h(t; \state{0}) &= \mathcal{J}_{\theta,\phi}(t) - \Bar{\mathcal{J}}(\state{0}).
\end{align}
The \emph{differential Bellman equation}~\cite{Todorov2009a} defines optimal behavior as recursively minimizing the mean-centered surprise at each time-step, or \emph{surprise-to-go}
\begin{align}
    \label{eq:differential_bellman}
    \Tilde{H}^{*}(t; \state{0}) &= h(t; \state{0}) + \min_{a_t} \expect{\state{t+1}\sim p_\theta(\cdot \mid \state{t})}{\Tilde{H}^{*}(t+1; \state{0})}.
\end{align}
The minimization over actions in Equation~\ref{eq:differential_bellman} assumes a fixed action space and feedforward planning, which may result in very high-dimensional recursive optimization problems. These assumptions also prove empirically, as well as computationally, problematic. Organisms are not born knowing all their affordances~\cite{Cisek2010}; they learn them~\cite{Pezzulo2016}. Noise~\cite{faisal2008noise,manohar2015reward}, uncertainty~\cite{gallivan2018decision}, and variability~\cite{scholz1999uncontrolled} are ubiquitous in motor control, and so movement must be stabilized by online feedback.

Stochastic optimal \emph{feedback} control therefore requires an optimality principle that allows for integrating observations between action steps. Rather than recursively optimize individual actions, Equation~\ref{eq:soft_differential_bellman} below therefore instead considers optimality of the feedback-stabilized transition density
\begin{align}
    \label{eq:soft_differential_bellman}
    \Tilde{H}^{*}(t; \state{0}) &= h(t; \state{0}) + \min_{q_\phi} \expect{\state{t+1}\sim q_\phi(\cdot \mid \state{t})}{\Tilde{H}^{*}(t+1; \state{0})}.
\end{align}
Equation~\ref{eq:soft_differential_bellman} defines an optimal controller as one that achieves optimal state transitions; individual actions act only as parameters to the optimal transition density. These optimal state transitions take the form of a generative model for agency, in which the generative model $p_\theta(\state{t+1} \mid \state{t})$ produces feasible state transitions and the Bellman optimality criterion ``weighs'' them according to their surprise-to-go
\begin{align}
    \label{eq:stochastic_reference}
    q^{*}(\state{t+1} \mid \state{t}) &= \frac{
        \expon{-\Tilde{H}^{*}(t+1; \state{0})} p_\theta(\state{t+1} \mid \state{t})
    }{
        \expect{
            \state{t+1} \sim p_\theta(\cdot \mid \state{t})
        }{
            \expon{-\Tilde{H}^{*}(t+1; \state{0})}
        }
    }.
\end{align}
The denominator of Equation~\ref{eq:stochastic_reference} would typically correspond to the marginal probability of an observation. Here it consists of the present state's expected surprise-to-go weight under the generative model. Potential future states that lead to high surprise under the reference model will have high surprise-to-go and therefore low weight under Equation~\ref{eq:stochastic_reference}. Present states that lead to states closely fitting the reference trajectory will have low surprise-to-go, resulting in a high denominator that spreads weight around among possible future states.

The availability of a closed-form density for the optimal transition density will help simplify the differential Bellman equation itself. Proposition~\ref{prop:path_integral_bellman} (in Appendix~\ref{app:derivations}) shows that by substituting Equation~\ref{eq:stochastic_reference} into Equation~\ref{eq:soft_differential_bellman} we can obtain a path-integral expression for the optimal differential surprise-to-go with both the feedforward controller $p_\theta$
\begin{align}
    \label{eq:pic_value_feedforward}
    \Tilde{H}^{*}(\state{0}) &= -\log \expect{
        p_\theta(\state{1:T} \mid \state{0})
    }{
        \expon{
            \sum_{t=1}^{T} \left(J(\state{t}) + L(\state{t}) \right) - \Bar{\mathcal{J}}(\state{0})
        }
    },
\end{align}
and the feedback controller $q_\phi$
\begin{align}
    \label{eq:pic_value_feedback}
    \Tilde{H}^{*}(\state{0}) &= -\log \expect{
        q_\phi(\state{1:T} \mid \state{0})
    }{
        \expon{
            \sum_{t=1}^{T} \mathcal{J}_{\theta,\phi}(t) - \Bar{\mathcal{J}}(\state{0})
        }
    }.
\end{align}
These equations employ ``smooth'' minimization rather than ``hard'' recursive minimization, and so they support feedforward planning, feedback-driven updating, and sensitivity of behavior to risk \cite{Theodorou2012,Pan2014}. Jensen's inequality will then yield a tractable upper bound on the optimal differential surprise-to-go under the feedback controller $q_{\phi}$
\begin{align}
    \label{eq:valbo}
    \Tilde{H}^{*}(\state{0})
    &\leq -\expect{
        q_\phi(\state{1:T} \mid \state{0})
    }{
        \sum_{t=1}^{T} h(t; \state{0})
    } = \Tilde{\mathcal{F}}^{*}_{\theta,\phi}.
\end{align}
Minimizing this \emph{differential free energy} \(\Tilde{\mathcal{F}}^{*}_{\theta,\phi}\) minimizes both the sensory surprise and the optimal surprise-to-go function by proxy. This kind of information-theoretic upper bound on a surprisal term is precisely what predictive coding process theories~\cite{bastos2012canonical,bogacz2017tutorial} posit that the brain can optimize by updating $\theta$ and $\phi$.

\section{Discussion}
\label{sec:apic_summary}

\paragraph{Related work} Our formulation follows in a tradition of unifying active inference with optimal control approaches.  Our hierarchical graphical model follows most closely from the one featured by Friston~\cite{Friston2017} and Pezzulo~\cite{Pezzulo2018} for hierarchical active inference in decision making and motor control. In contrast to theirs, our model includes only a single observation at the lowest hierarchical level rather than one observed variable per level.

We also draw inspiration from information-theoretic control schemes not labelled by their authors as ``active inference''. Piray and Daw~\cite{piray2021linear} considered a path-integral control approach to planning and reinforcement learning, which they related to grid cells in the entorhinal cortex. Mitchell et al~\cite{mitchell2019minimum} modeled motor learning as minimization of a free energy functional. Nasriany et al's work on distribution-conditioned reinforcement learning gave us our scheme for parameterizing reference distributions \cite{Nasiriany2021}, and Sennesh et al~\cite{Sennesh2021} applied such an objective to active inference modeling of interoception and allostatic regulation.

\paragraph{Implementations} We employed the infinite-horizon, average-surprise criterion to fit with the apparent time-averaging of dopamine signals in the brain \cite{daw2000behavioral,shadmehr2020vigor}, but algorithms for this control criterion remain an active research area with no standard approach. A recent survey~\cite{Lanillos2021} showed that most software implementations of active inference models still involve either finite horizons or exponential discounting criteria. Those which do support infinite horizons and nonlinear model families mostly take algorithmic inspiration from reinforcement learning (RL).

In that domain, Tadepalli and Ok~\cite{Tadepalli1998} published the first model-based RL algorithm for our criterion in 1998, while Baxter and Bartlett~\cite{Baxter2001} gave a biased policy gradient estimator. It took another decade for Alexander and Brown~\cite{Alexander2010} to give a recursive decomposition for average-cost temporal-difference learning. Zhang and Ross~\cite{Zhang2021} have only recently published the first adaptation of ``deep'' reinforcement learning algorithms (based on function approximation) to the average-cost criterion, which remains model free. Jafarnia-Jahromi et al~\cite{Jafarnia-Jahromi2022} recently gave the first algorithm for infinite-horizon, average-cost partially observable problems with a known observation density and unknown dynamics.

\paragraph{Conclusion} This concludes the derivation of an infinite-horizon, average-surprise formulation of active inference. Since our formulation contextualizes behavioral episodes, it only requires planning and adjusting behavior in context (e.g.~from timesteps $1$ to $T$), despite optimizing a ``global'' (indefinite) surprise rate (Equation~\ref{eq:indefinite_capture_rate}). We suggest that this formulation of active inference can advance a probabilistic approach to model-based, hierarchical feedback control~\cite{Pezzulo2016,Merel2019}.

\appendix

\section{Detailed derivations}
\label{app:derivations}

This appendix provides detailed derivations for equations used elsewhere, particularly where doing so would have distracted from the flow of the paper.

\begin{proposition}[Variational free energy as divergence from an unnormalized joint distribution]
\label{prop:vfe}
The variational free energy (Equation~\ref{eq:vfe}) is defined as the Kullback-Leibler divergence of the recognition model $q_\phi$ from the unnormalized joint distribution of the generative model $p_\theta$
\begin{align*}
    \mathcal{F}_{\theta, \phi}(t) &= \KL{
        q_\phi(s^{(1:2)}_{t}, a^{(1:2)}_{t} \mid o_{t}, \state{t+1}, \state{t-1})
    }{
        p_\theta(\state{t} \mid \state{t-1})
    },
\end{align*}
and therefore equals a sum of the cross entropy between the recognition model and the sensory likelihood and the exclusive KL divergence from the recognition model to the generative model over the latent variables
\begin{multline*}
    \mathcal{F}_{\theta, \phi}(t) = \expect{q_\phi}{- \log p_{\theta}(o_t \mid a^{(1)}_t, s^{(1)}_t)} + \\ \KL{q_\phi(s^{(1:2)}_{t}, a^{(1:2)}_{t} \mid o_{t}, \state{t+1}, \state{t-1})}{p_{\theta}(s^{(1:2)}_t, a^{(1:2)}_t \mid \state{t-1})}.
\end{multline*}
\end{proposition}
\begin{proof}
Taking a divergence between the (normalized) recognition model and the (unnormalized) joint generative model will yield
\begin{align*}
    \mathcal{F}_{\theta, \phi}(t) &= \KL{
        q_\phi(s^{(1:2)}_{t}, a^{(1:2)}_{t} \mid o_{t}, \state{t+1}, \state{t-1})
    }{
        p_\theta(\state{t} \mid \state{t-1})
    } \\
    &= \expect{
        q_\phi(s^{(1:2)}_{t}, a^{(1:2)}_{t} \mid o_{t}, \state{t+1}, \state{t-1})
    }{
        -\log \frac{
            p_\theta(\state{t} \mid \state{t-1})
        }{
            q_\phi(s^{(1:2)}_{t}, a^{(1:2)}_{t} \mid o_{t}, \state{t+1}, \state{t-1})
        }
    } \\
    &= \expect{
        q_\phi(s^{(1:2)}_{t}, a^{(1:2)}_{t} \mid o_{t}, \state{t+1}, \state{t-1})
    }{
        - \log \frac{
            p_{\theta}(o_t \mid a^{(1)}_t, s^{(1)}_t)
            p_{\theta}(s^{(1:2)}_t, a^{(1:2)}_t \mid \state{t-1})
        }{
            q_\phi(s^{(1:2)}_{t}, a^{(1:2)}_{t} \mid o_{t}, \state{t+1}, \state{t-1})
        }
    } \\
    &= \expect{q_\phi}{- \log p_{\theta}(o_t \mid a^{(1)}_t, s^{(1)}_t)} -
    \expect{
        q_\phi
    }{
        \log \frac{
            p_{\theta}(s^{(1:2)}_t, a^{(1:2)}_t \mid \state{t-1})
        }{
            q_\phi(s^{(1:2)}_{t}, a^{(1:2)}_{t} \mid o_{t}, \state{t+1}, \state{t-1})
        }
    },
\end{align*}
as required.
\end{proof}

\begin{proposition}[KL divergence of the optimal feedback controller from the feedforward controller]
\label{prop:optimal_kl}
The exclusive Kullback-Leibler divergence of the optimal feedback controller $q^{*}$ from the feedforward generative model $p_\theta$ is
\begin{multline}
    \label{eq:optimal_kl}
    \KL{q^{*}(\state{t+1} \mid \state{t})}{p_\theta(\state{t+1} \mid \state{t})}
    =
    -\expect{q^{*}(\state{t+1} \mid \state{t})}{
        \Tilde{H}^{*}(t+1; \state{0})
    } - \\
    \log \expect{p_\theta(\state{t+1} \mid \state{t})}{\expon{-\Tilde{H}^{*}(t+1; \state{0})}}.
\end{multline}
\end{proposition}
\begin{proof}
We begin by writing out the definition of a KL divergence
\begin{align*}
    \KL{q^{*}(\state{t+1} \mid \state{t})}{p_\theta(\state{t+1} \mid \state{t})}
    &= \expect{q^{*}(\state{t+1} \mid \state{t})}{
        -\log \frac{
            p_\theta(\state{t+1} \mid \state{t})
        }{
            q^{*}(\state{t+1} \mid \state{t})
        }
    }.
\end{align*}
The definition of $q^{*}$ in terms of $p_\theta$ (Equation~\ref{eq:stochastic_reference}) allows the inner ratio of densities to simplify to
\begin{align*}
    \frac{
        p_\theta(\state{t+1} \mid \state{t})
    }{
        q^{*}(\state{t+1} \mid \state{t})
    } &= p_\theta(\state{t+1} \mid \state{t}) \left(q^{*}(\state{t+1} \mid \state{t}) \right)^{-1} \\
    &= \cancel{p_\theta(\state{t+1} \mid \state{t})} \left(\frac{
        \expect{p_\theta(\state{t+1} \mid \state{t})}{\expon{-\Tilde{H}^{*}(t+1; \state{0})}}
    }{
        \expon{-\Tilde{H}^{*}(t+1; \state{0})} \cancel{p_\theta(\state{t+1} \mid \state{t})}
    }\right) \\
    \frac{
        p_\theta(\state{t+1} \mid \state{t})
    }{
        q^{*}(\state{t+1} \mid \state{t})
    }
    &= \frac{
        \expect{p_\theta(\state{t+1} \mid \state{t})}{\expon{-\Tilde{H}^{*}(t+1; \state{0})}}
    }{
        \expon{-\Tilde{H}^{*}(t+1; \state{0})}
    }.
\end{align*}
This simplified ratio therefore has the logarithm
\begin{align*}
    \log \frac{
        p_\theta(\state{t+1} \mid \state{t})
    }{
        q^{*}(\state{t+1} \mid \state{t})
    }
    &=
    \log \expect{p_\theta(\state{t+1} \mid \state{t})}{\expon{-\Tilde{H}^{*}(t+1; \state{0})}} + \Tilde{H}^{*}(t+1; \state{0})
\end{align*}
and the divergence becomes
\begin{multline*}
    \KL{q^{*}(\state{t+1} \mid \state{t})}{p_\theta(\state{t+1} \mid \state{t})}
    = \\
    -\expect{q^{*}(\state{t+1} \mid \state{t})}{
        \Tilde{H}^{*}(t+1; \state{0})
    } - \log \expect{p_\theta(\state{t+1} \mid \state{t})}{\expon{-\Tilde{H}^{*}(t+1; \state{0})}}.
\end{multline*}
\end{proof}

\begin{proposition}[Path-integral expression for the optimal differential surprise-to-go]
\label{prop:path_integral_bellman}
The optimal differential surprise-to-go function defined by the Bellman equation (Equation~\ref{eq:soft_differential_bellman})
\begin{align*}
    \Tilde{H}^{*}(t; \state{0}) &= h(t; \state{0}) + \min_{q_\phi} \expect{\state{t+1}\sim q_\phi(\cdot \mid \state{t})}{\Tilde{H}^{*}(t+1; \state{0})}
\end{align*}
can be simplified by substituting in $q^{*}$ to obtain a path-integral expression
\begin{align*}
    \Tilde{H}^{*}(\state{0}) &= -\log \expect{
        p_\theta(\state{1:T} \mid \state{0})
    }{
        \expon{
            \sum_{t=1}^{T} \left(J(\state{t}) + L(\state{t}) \right) - \Bar{\mathcal{J}}(\state{0})
        }
    }, \\
    &= -\log \expect{
        q_\phi(\state{1:T} \mid \state{0})
    }{
        \expon{
            \sum_{t=1}^{T} \mathcal{J}_{\theta,\phi}(t) - \Bar{\mathcal{J}}(\state{0})
        }
    }.
\end{align*}
\end{proposition}
\begin{proof}
Substituting Equation~\ref{eq:stochastic_reference} into Equation~\ref{eq:soft_differential_bellman} yields
\begin{align}
    \label{eq:optimal_differential_bellman}
    \Tilde{H}^{*}(t; \state{0}) &= \Bar{\mathcal{J}}(\state{0}) - \mathcal{J}_{\theta,\phi}(t) + \expect{q^{*}(\state{t+1} \mid \state{t})}{\Tilde{H}^{*}(t+1; \state{0})},
\end{align}
whose recursive term is $\expect{q^{*}(\state{t+1} \mid \state{t})}{\Tilde{H}^{*}(t+1; \state{0})}$. The divergence term in $\mathcal{J}$ (Equation~\ref{eq:pic_objective_long}) will cancel this term. By Proposition~\ref{prop:optimal_kl} the divergence equals
\begin{multline*}
    \KL{q^{*}(\state{t+1} \mid \state{t})}{p_\theta(\state{t+1} \mid \state{t})}
    = \\
    -\expect{q^{*}(\state{t+1} \mid \state{t})}{
        \Tilde{H}^{*}(t+1; \state{0})
    } - \log \expect{p_\theta(\state{t+1} \mid \state{t})}{\expon{-\Tilde{H}^{*}(t+1; \state{0})}}.
\end{multline*}
Substituting Equation~\ref{eq:optimal_kl} into Equation~\ref{eq:pic_objective_long} will yield
\begin{multline*}
    -\mathcal{J}_{\theta,\phi}(t) = \expect{q^{*}(\state{t+1} \mid \state{t})}{\Tilde{H}^{*}(t+1; \state{0})} + \log \expect{p_\theta(\state{t+1} \mid \state{t})}{\expon{-\Tilde{H}^{*}(t+1; \state{0})}} \\
    + \expect{q_\phi}{-J(\state{t})} + \expect{q_\phi}{-L(\state{t})},
\end{multline*}
whose first term will cancel the recursive optimization when substituted into Equation~\ref{eq:optimal_differential_bellman}. The result will be a ``smoothly minimizing'' expression for the optimal differential surprise-to-go
\begin{multline*}
    \Tilde{H}^{*}(t; \state{0}) = \Bar{\mathcal{J}}(\state{0}) - \left(J(\state{t}) + L(\state{t})\right) \\ -\log \expect{p_\theta(\state{t+1} \mid \state{t})}{\expon{-\Tilde{H}^{*}(t+1; \state{0})}},
\end{multline*}
and after unfolding of the recursive expectation, a path-integral expression for the optimal differential surprise-to-go
\begin{align*}
    \Tilde{H}^{*}(\state{0}) &= -\log \expect{
        p_\theta(\state{1:T} \mid \state{0})
    }{
        \expon{
            \sum_{t=1}^{T} \left(J(\state{t}) + L(\state{t}) \right) - \Bar{\mathcal{J}}(\state{0})
        }
    }.
\end{align*}
Sampling a trajectory of states from a feedback controller $q_\phi$ instead of the feedforward planner $p_\theta$ will then result in a nonzero divergence term
\begin{align*}
    \Tilde{H}^{*}(\state{0}) &= -\log \expect{
        q_\phi(\state{1:T} \mid \state{0})
    }{
        \expon{
            \sum_{t=1}^{T} \mathcal{J}_{\theta,\phi}(t) - \Bar{\mathcal{J}}(\state{0})
        }
    }.
\end{align*}
\end{proof}

%
%
%
\bibliographystyle{splncs04}
\bibliography{main}

\end{document}